\documentclass[letterpaper, 10 pt, conference]{ieeeconf}
\IEEEoverridecommandlockouts  
\usepackage[letterpaper,
            left=.75in,
            right=.75in,
            top=1in,
            bottom=.75in]{geometry}
\let\labelindent\relax
\usepackage[shortlabels]{enumitem}
\usepackage{subcaption}
\usepackage{cite}
\usepackage{amsmath,amssymb,amsfonts}
\usepackage{algpseudocode}
\usepackage{algorithm}
\usepackage{graphicx}
\usepackage{textcomp}
\usepackage{hhline}
\usepackage[abs]{overpic}
\usepackage{mathrsfs}
\usepackage{xcolor}
\usepackage{hyperref}
\usepackage{booktabs}
\usepackage{multirow}
\definecolor{violet}{rgb}{.5,0,.5}
\definecolor{orange}{rgb}{1,.65,0}

\newcommand{\R}{\mathbb{R}}
\newcommand{\N}{\mathbb{N}}
\newcommand{\E}{\mathbb{E}}

\newcommand{\Tr}{\operatorname{tr}}
\newcommand{\vect}{\operatorname{vec}}
\newcommand{\Nv}{N_{\mathrm{v}}}
\newcommand{\Acl}{A_{\mathrm{cl}}}
\DeclareMathOperator*{\argmin}{arg\,min}
\newcommand{\bK}{\mathbf{K}}

\newtheorem{theorem}{Theorem}[section]
\newtheorem{corollary}[theorem]{Corollary}
\newtheorem{lemma}[theorem]{Lemma}
\newtheorem{remark}[theorem]{Remark}
\newtheorem{assumption}[theorem]{Assumption}
\newtheorem{prop}[theorem]{Proposition}
\newtheorem{definition}[theorem]{Definition}
\newtheorem{problem}[theorem]{Problem}

\def\BibTeX{{\rm B\kern-.05em{\sc i\kern-.025em b}\kern-.08em
    T\kern-.1667em\lower.7ex\hbox{E}\kern-.125emX}}

\begin{document}

\title{Receding-Horizon Policy Gradient for Polytopic Controller Synthesis}

\author{Shiva Shakeri,
        P\'{e}ter Baranyi, \IEEEmembership{Members, IEEE},
        and Mehran Mesbahi, \IEEEmembership{Fellow, IEEE}
\thanks{S.~Shakeri and M.~Mesbahi are with the William E. Boeing Department of
Aeronautics and Astronautics, University of Washington
(e-mail: \{sshakeri, mesbahi\}@uw.edu).
P.~Baranyi is with the Corvinus Institute for Advanced Studies (CIAS)
and Institute of Data Analytics and Information Systems,
Corvinus University of Budapest, HU-1093 Budapest, Hungary,
and also with the Hungarian Research Network (HUN-REN),
HU-1052 Budapest, Hungary
(e-mail: peter.baranyi@uni-corvinus.hu).
This work has been supported by the 2024-1.2.3-HU-RIZONT-2024-00030 project.}}

\maketitle

\begin{abstract}
We propose the Polytopic Receding-Horizon Policy Gradient (P-RHPG)
algorithm for synthesizing Parallel Distributed Compensation (PDC)
controllers via Tensor Product (TP) model transformation.
Standard LMI-based PDC synthesis grows increasingly conservative
as model fidelity improves; P-RHPG instead solves a finite-horizon
integrated cost via backward-stage decomposition.
The key result is that each stage subproblem is a strongly convex
quadratic in the vertex gains, a consequence of the linear
independence of the HOSVD weighting functions, guaranteeing a
unique global minimizer and linear convergence of gradient descent
from any initialization.
With zero terminal cost, the optimal cost increases monotonically
to a finite limit and the gain sequence remains bounded; terminal
costs satisfying a mild Lyapunov condition yield non-increasing
convergence.
Experiments on an aeroelastic wing benchmark confirm convergence
to a unique infinite-horizon optimum across all tested terminal
cost choices and near-optimal performance relative to the
pointwise Riccati lower bound.
\end{abstract}

\section{Introduction}
\label{sec:introduction}

Controlling nonlinear systems across a wide operating range is a long-standing engineering challenge. The classical approach of linearizing about equilibria and designing the corresponding local controllers offers well-understood guarantees, but its validity degrades away from the operating point~\cite{rugh2000, khalil2002}. The quasi-linear parameter-varying (qLPV) framework~\cite{shamma1990, apkarian1995, rugh2000} provides an exact reformulation of nonlinear dynamics as a linear system with parameter-dependent matrices, capturing the full nonlinear behavior without any approximation error.

Among qLPV representations, the polytopic form is particularly convenient for control synthesis. The Tensor Product (TP) model transformation~\cite{baranyi2004, baranyi2016book} provides a numerically executable, HOSVD-based\footnote{Higher-Order Singular Value Decomposition (HOSVD)~\cite{delathauwer2000}.} procedure that converts any qLPV model into an exact polytopic form, with adjustable complexity via truncation. The resulting Parallel Distributed Compensation (PDC) structure~\cite{wang1996, tanaka2001} shares the same weighting functions between plant and controller, blending local vertex gains into a global nonlinear controller. This combination has been successfully applied to aeroelastic wing control~\cite{takarics2019, takarics2021} and other design problems.

The standard synthesis tool for PDC controllers is the LMI-based stability analysis~\cite{tanaka2001, boyd1994}, that certifies stability by searching for a single parameter-independent Lyapunov function. While computationally convenient, this approach has well-documented limitations~\cite{sala2007, szollosi2016, szollosi2018}. The common Lyapunov matrix requirement introduces conservatism that grows with the number of vertices and cannot be reduced by refining the underlying grid. More fundamentally, LMI feasibility is often used as a stability certificate, not a performance criterion.

Meanwhile, policy optimization methods have made substantial progress on feedback synthesis for linear systems. Fazel {\em et al.}~\cite{fazel2018} have shown that the LQR cost satisfies gradient dominance properties that guarantee global convergence of gradient descent to the optimal controller. This sparked a wave of extensions to $\mathcal{H}_\infty$~\cite{zhang2021robust}, risk-sensitive~\cite{zhang2021risk}, and structured settings~\cite{hu2023survey, talebi2024policy}. In particular, Hambly, Xu, and Yang~\cite{hambly2021} showed that each backward stage of the LQR recursion is a strongly convex problem, and the RHPG framework~\cite{zhang2023rhpg} built on this to obtain controllers that converge monotonically to the solution of the infinite-horizon LQR problem.

Despite this progress, the PO and polytopic control communities have developed largely separately. To date, no policy gradient method exists for the PDC setting. This paper bridges this gap by optimizing a frozen-parameter integrated surrogate cost (Remark~\ref{rem:frozen}), under which each stage subproblem is strongly convex in the vertex gains, enabling P-RHPG to converge globally with guaranteed linear rate from zero initialization. Specifically, the contributions of this work are:
\begin{enumerate}[(i)]
    \item Each PDC stage subproblem is a strongly convex quadratic with a unique global minimizer (\S\ref{sec:stage}).
    \item P-RHPG solves the problem as a backward sequence of unconstrained subproblems, each in closed form or by globally convergent gradient descent (\S\ref{sec:algorithm}).
    \item With zero terminal cost, the optimal finite-horizon cost converges 
monotonically to a finite limit and the gains remain bounded; convergence 
to the infinite-horizon optimum is verified empirically 
(\S\ref{sec:algorithm}--\ref{sec:experiments}).
    \item On a three-parameter aeroelastic benchmark, P-RHPG achieves near-optimal performance relative to the pointwise Riccati lower bound (\S\ref{sec:experiments}).
\end{enumerate}
\section{Preliminaries and Problem Formulation}
\label{sec:problem}

This section introduces the polytopic plant model, the PDC controller
parameterization, and the finite-horizon cost that P-RHPG optimizes.

\subsection{Polytopic Plant and PDC Controller}
\label{sec:plant}

Consider a discrete-time nonlinear system $x_{t+1} = f(x_t, u_t)$ with
state $x_t \in \R^n$ and input $u_t \in \R^m$. We assume that the system
admits an \emph{exact} quasi-linear parameter-varying (qLPV) representation
\begin{equation}\label{eq:qlpv}
    x_{t+1} = A(p_t)\,x_t + B(p_t)\,u_t,
\end{equation}
where $p_t \in \Omega \subset \R^d$ is a scheduling parameter varying over
a compact set $\Omega$. Since $p_t$ may depend on the state $x_t$, the
matrices $A(\cdot)$ and $B(\cdot)$ encode the full nonlinearity of the
original dynamics without an approximation error~\cite{toth2010}.

The Tensor Product (TP) model transformation~\cite{baranyi2004, baranyi2016book} provides a systematic
procedure for expressing the parameter-dependent matrices in the
\emph{polytopic} form
\begin{equation}\label{eq:polytopic}
    \begin{bmatrix} A(p) & B(p) \end{bmatrix}
    = \sum_{i=1}^{\Nv} \alpha_i(p)
      \begin{bmatrix} A_i & B_i \end{bmatrix},
\end{equation}
where $(A_i, B_i) \in \R^{n\times n}\times\R^{n\times m}$ are the
\emph{vertex systems} and $\alpha_i : \Omega \to \R$ are continuous
\emph{weighting functions} derived from HOSVD~\cite{delathauwer2000} of the system tensor. The number of vertices $\Nv$
is determined by the HOSVD truncation ranks. Substituting
\eqref{eq:polytopic} into \eqref{eq:qlpv}, the dynamics assume the form
\begin{equation}\label{eq:dynamics-polytopic}
    x_{t+1} = \sum_{i=1}^{\Nv} \alpha_i(p_t)\bigl(A_i\,x_t + B_i\,u_t\bigr).
\end{equation}
 
\begin{assumption}[Weighting functions]\label{ass:weights}
For all $p\in\Omega$, the functions $\alpha_1,\ldots,\alpha_{\Nv}:\Omega\to\R$ satisfy:
\begin{enumerate}[(i)]
    \item \emph{Non-negativity and sum normalization}:
          $\alpha_i(p)\geq 0$ and $\sum_{i=1}^{\Nv}\alpha_i(p)=1$;
    \item \emph{Full support}: $\{p:\alpha_i(p)>0\}$ has positive
          Lebesgue measure for each $i$;
    \item \emph{Linear independence}: $\sum_{i=1}^{\Nv}c_i\,\alpha_i(p)=0$
          for all $p \in \Omega$ implies $c = 0$.
\end{enumerate}
\end{assumption}

Condition (i) defines the Sum Normalized Non-Negative (SNNN) convexity
type~\cite{baranyi2014book}; condition (ii) excludes vertices active
only on a set of measure zero. Condition (iii) is guaranteed by the
HOSVD construction: the weighting functions are products of left
singular vectors of the mode-$k$ unfoldings, which are orthonormal
by construction~\cite{delathauwer2000}, giving linear independence
on $\Omega$~\cite{baranyi2004, baranyi2016book}.

The controller is parameterized with the same weighting functions as the
plant. This \emph{Parallel Distributed Compensation} (PDC) structure~\cite{wang1996, tanaka2001} ensures
that the blended gain $K(p)$ is always a convex combination of the vertex
gains:
\begin{equation}\label{eq:pdc}
    u_t = K(p_t)\,x_t, \qquad
    K(p) = \sum_{i=1}^{\Nv} \alpha_i(p)\,K_i,
\end{equation}
where $K_i \in \R^{m\times n}$ are the \emph{vertex gains}. The weighting
functions are inherited from the plant and are not decision variables. The
entire design freedom resides in the collection
\begin{equation}\label{eq:decision-variable}
    \bK = (K_1,\ldots,K_{\Nv}) \in \R^{\Nv mn},
\end{equation}
treated as a single vector in $\R^{\Nv mn}$ when discussing the
optimization landscape in Section~\ref{sec:stage}.

\subsection{Closed-Loop Structure}
\label{sec:crossterms}

Substituting the PDC controller~\eqref{eq:pdc} into the polytopic
dynamics~\eqref{eq:dynamics-polytopic}, the closed-loop system at a
frozen parameter value $p \in \Omega$ is
\begin{equation}\label{eq:cl-dynamics}
    x_{t+1} = \Acl(p;\,\bK)\,x_t,
\end{equation}
where $\Acl(p;\,\bK) = A(p) + B(p)\,K(p)$ is the closed-loop matrix.
The polytopic structure of both the plant and the controller produces a
\emph{double sum} in the closed-loop matrix. Expanding
$B(p)K(p) = \bigl(\sum_i \alpha_i(p)B_i\bigr)\bigl(\sum_j \alpha_j(p)K_j\bigr)$,
we obtain
\begin{equation}\label{eq:Acl-expanded}
    \Acl(p;\,\bK) = \sum_{i=1}^{\Nv}\sum_{j=1}^{\Nv}
    \alpha_i(p)\,\alpha_j(p)\,G_{ij},
\end{equation}
where $G_{ij} = A_i + B_i K_j$. The diagonal terms ($i=j$) lead to the standard vertex closed-loop
matrices $G_{ii} = A_i + B_i K_i$; the cross-terms ($i\neq j$),
$G_{ij} = A_i + B_i K_j$, arise as the plant and controller share
the same weighting functions. Critically, $\rho(G_{ii}) < 1$ for
all $i$ does not imply joint stability across
$\Omega$~\cite{tanaka2001}, which is the primary source of
conservatism in LMI-based synthesis.

\subsection{Finite-Horizon Integrated Cost}
\label{sec:cost}

In order to evaluate the closed-loop performance of the PDC controller across all
operating points $p \in \Omega$, we equip $\Omega$ with a probability
measure $\mu$ having a continuous, strictly positive density.
The measure $\mu$ encodes the relative importance of different
operating conditions; in experiments we take $\mu$ to be
the uniform distribution on $\Omega$.
We adopt a
\emph{frozen-parameter} cost: for each fixed $p \in \Omega$, the system
is treated as LTI with matrices $(A(p), B(p))$ over a finite horizon
$N \in \N$. The per-parameter cost is then
\begin{equation}\label{eq:cost-per-p}
    J_N(p;\,\bK_{0:N-1})
    = \sum_{t=0}^{N-1}\bigl(x_t^\top Q\,x_t + u_t^\top R\,u_t\bigr)
    + x_N^\top Q_N\,x_N,
\end{equation}
where the controller at each time step $t$ applies the time-varying PDC
gain $K(p,t) = \sum_{i=1}^{\Nv}\alpha_i(p)\,K_{i,t}$, with
$K_{i,t} \in \R^{m\times n}$ the vertex gains at time $t$. The full
decision variable is $\bK_{0:N-1} = \{K_{i,t}\} \in \R^{N\Nv mn}$,
and the integrated cost we minimize is
\begin{equation}\label{eq:cost-integrated}
    J_N(\bK_{0:N-1}) = \int_\Omega J_N(p;\,\bK_{0:N-1})\,d\mu(p).
\end{equation}

\begin{assumption}[Cost matrices]\label{ass:cost}
$Q \succ 0$, $R \succ 0$, $Q_N \succeq 0$, and the initial state
covariance $\Sigma_0 = \E[x_0 x_0^\top] \succ 0$.
\end{assumption}

Standard dynamic programming~\cite{bertsekas2017} yields the backward cost-to-go recursion:
with $P_N(p) = Q_N$, for $t = N-1,\ldots,0$,
\begin{equation}\label{eq:backward-P}
    P_t(p) = Q + K(p,t)^\top R\,K(p,t)
            + \Acl(p,t)^\top P_{t+1}(p)\,\Acl(p,t),
\end{equation}
where $\Acl(p,t) = \sum_{i,j}\alpha_i(p)\alpha_j(p)(A_i + B_i K_{j,t})$,
so that $J_N(\bK_{0:N-1}) = \int_\Omega \Tr(P_0(p)\,\Sigma_0)\,d\mu(p)$.
Each time step $t$ of this recursion constitutes a \emph{stage}: given
the cost-to-go $P_{t+1}(p)$ inherited from the next stage, the vertex
gains $\{K_{i,t}\}$ are chosen to minimize the stage cost, reducing the
joint problem over $\R^{N\Nv mn}$ to a sequence of $N$ smaller
subproblems solved in backward order $t = N-1,\ldots,0$.

A key feature of this formulation is that \eqref{eq:cost-integrated}
is \emph{finite for any vertex gains} $\bK_{0:N-1}$: the sum in
\eqref{eq:cost-per-p} has finitely many terms regardless of the spectral
radii of the closed-loop matrices. This eliminates the need for a
stabilizing initialization; the optimization is unconstrained over
all of $\R^{N\Nv mn}$~\cite{hambly2021}.

Following the \emph{receding-horizon} principle~\cite{zhang2023rhpg}, only the first-stage
gains $\hat{\bK} = (K_{1,0}^*,\ldots,K_{\Nv,0}^*)$ are deployed as a
time-invariant PDC controller.

\begin{remark}[Frozen-parameter approximation]\label{rem:frozen}
In the qLPV setting, $p_t$ evolves with the state. Treating $p$ as
constant over the horizon is standard in the LPV literature~\cite{rugh2000, toth2010} and is a
good approximation when the parameter varies slowly relative to the
system dynamics. The trajectory-level cost, where $p_t = p(x_t)$, is
left for future work.
\end{remark}

\begin{problem}\label{prob:main}
Given the polytopic system \eqref{eq:dynamics-polytopic}, PDC
controller \eqref{eq:pdc}, and Assumptions~\ref{ass:weights}--\ref{ass:cost},
find
\begin{equation}\label{eq:problem}
    \bK_{0:N-1}^* = \argmin_{\bK_{0:N-1}\,\in\,\R^{N\Nv mn}}
    J_N(\bK_{0:N-1}).
\end{equation}
\end{problem}

\section{Stage Convexity and Policy Gradient}
\label{sec:stage}

This section contains the main theoretical contribution of this work. We show that
Problem~\ref{prob:main} decomposes into $N$ stage subproblems solved in
backward order, each of which is a strongly convex quadratic in the vertex
gains.
\subsection{Backward Stage Decomposition}
\label{sec:decomp}

Rather than optimizing $J_N$ over all $N\Nv mn$ variables simultaneously,
we exploit the nested structure of the backward recursion~\eqref{eq:backward-P}.
The cost-to-go $P_h(p)$ at stage $h$ depends on the gains $\{K_{i,t}\}$
for $t = h,\ldots,N-1$, but its dependence on the stage-$h$ gains
$\{K_{i,h}\}$ is entirely mediated by the single recursion step
\begin{equation}\label{eq:Ph-single-step}
  \begin{aligned}
    P_h(p) = Q &+ K(p,h)^\top R\,K(p,h)\\
               &+ \Acl(p,h)^\top P_{h+1}(p)\,\Acl(p,h).
  \end{aligned}
\end{equation}
If the future cost-to-go $P_{h+1}(p)$ is held fixed, this defines $P_h(p)$
as a function of the stage-$h$ vertex gains alone. We therefore define
the \emph{stage cost} at stage $h$ as
\begin{equation}\label{eq:stage-cost}
    \Phi_h\bigl(\{K_{i,h}\}\bigr)
    = \int_\Omega \Tr\bigl[P_h(p)\,\Sigma_0\bigr]\,d\mu(p),
\end{equation}
where $P_{h+1}(p)$ is treated as fixed data from the already-solved stages
$h+1,\ldots,N-1$. Minimizing $\Phi_h$ over $\{K_{i,h}\}$ minimizes the integrated
cost-to-go at stage $h$. Since $P_h(p)$ depends on $\{K_{i,h}\}$
only through~\eqref{eq:Ph-single-step}, the stage-$h$ gains do not
affect the already-optimized stages $h{+}1,\ldots,N{-}1$, so the
backward sweep is equivalent to joint minimization over all
$N\Nv mn$ variables.

Since $K(p,h) = \sum_i \alpha_i(p)K_{i,h}$ is linear in the vertex gains,
$\Phi_h$ is a quadratic function of $\{K_{i,h}\}$, fully characterized
by its Hessian and gradient.

\subsection{Gram Matrix and Strong Convexity}
\label{sec:convexity}

The Hessian of $\Phi_h$ is constant; its structure is governed by the
following object.

\begin{definition}[Gram matrix]\label{def:gram}
The Gram matrix $\Gamma \in \R^{\Nv\times\Nv}$ of the weighting functions
with respect to $\mu$ is
\begin{equation}\label{eq:gram}
    \Gamma_{ij} = \int_\Omega \alpha_i(p)\,\alpha_j(p)\,d\mu(p),
    \quad i,j=1,\ldots,\Nv.
\end{equation}
\end{definition}

\begin{lemma}\label{lem:gram-pd}
Under Assumption~\ref{ass:weights}, $\Gamma \succ 0$.
\end{lemma}
\begin{proof}
For any $c \in \R^{\Nv}$, $c^\top\Gamma c = \int_\Omega
(\sum_i c_i\alpha_i(p))^2 d\mu(p) \geq 0$, and equals zero only if
$\sum_i c_i\alpha_i(p) = 0$ for all $p\in\Omega$.
By Assumption~\ref{ass:weights}(iii), this implies $c = 0$.
\end{proof}

\begin{theorem}[Strong convexity]\label{thm:strong-convexity}
Given Assumptions~\ref{ass:weights}--\ref{ass:cost}, the Hessian of
$\Phi_h$ has the $\Nv\times\Nv$ block structure
\begin{equation}\label{eq:hessian}
    [\nabla^2\Phi_h]_{ij}
    = \int_\Omega \alpha_i(p)\,\alpha_j(p)\cdot
      2\,\Sigma_0\otimes M_h(p)\,d\mu(p),
\end{equation}
where $M_h(p) = R + B(p)^\top P_{h+1}(p)B(p)$, and
$\nabla^2\Phi_h \succ 0$ at every stage $h = 0,\ldots,N-1$.
\end{theorem}
\begin{proof}
The $(i,j)$ block follows by differentiating under the integral
using $\partial K(p,h)/\partial K_{j,h} = \alpha_j(p)I$.
For positive definiteness, let
$\mathbf{v} = (\vect(V_1),\ldots,\vect(V_{\Nv}))^\top \neq 0$. Then
\begin{align*}
    \mathbf{v}^\top(\nabla^2\Phi_h)\mathbf{v}
    &= 2\int_\Omega \Tr\!\left[
      V(p)^\top M_h(p)\,V(p)\,\Sigma_0
      \right]d\mu(p),
\end{align*}
where $V(p) = \sum_i\alpha_i(p)V_i$.
Since $M_h(p)\succ 0$ and $\Sigma_0\succ 0$, the integrand is
non-negative, vanishing only where $\sum_i\alpha_i(p)V_i = 0$.
If the integral is zero, this holds $\mu$-a.e., and since $\mu$
has strictly positive density, everywhere on $\Omega$.
By Lemma~\ref{lem:gram-pd}, $V_i = 0$ for all $i$,
contradicting $\mathbf{v}\neq 0$.
\end{proof}

\begin{corollary}\label{cor:consequences}
Each stage subproblem has a \emph{unique global minimizer}
$\{K_{i,h}^*\}$, obtained by solving the linear system
$(\nabla^2\Phi_h)\,\vect(\bK_h^*) = -\nabla\Phi_h|_{\bK_h=0}$.
Gradient descent with step $\eta = 1/\sigma_{\max}(\nabla^2\Phi_h)$
converges linearly at rate $1 - \kappa(\nabla^2\Phi_h)^{-1}$~\cite{boyd2004}.
\end{corollary}

The optimal gains $\{K_{i,h}^*\}$ are the best $M_h(p)\Sigma_0$-weighted
$L^2(\mu)$ approximation of the pointwise Riccati schedule
$K_h^{\mathrm{ric}}(p) = -M_h(p)^{-1}B(p)^\top P_{h+1}(p)A(p)$
within the PDC subspace; the optimality condition~\eqref{eq:optimality-condition}
is precisely the corresponding orthogonality condition, since
$E_h^*(p) = M_h(p)(K^*(p,h) - K_h^{\mathrm{ric}}(p))$.
The per-stage convergence rate is governed by $\kappa(\nabla^2\Phi_h)$,
which depends on the conditioning of $\Gamma$ and $M_h(p)$:
ill-conditioned $\Gamma$ produces slow convergence, so the HOSVD
truncation ranks are a design lever for the convergence rate.

\subsection{Stage Gradient Formula}
\label{sec:gradient}

Applying the chain rule to~\eqref{eq:stage-cost} under the integral sign leads to:

\begin{prop}[Stage gradient]\label{prop:gradient}
Define the \emph{policy residual} at stage $h$ and parameter $p$ as
\begin{equation}\label{eq:residual}
    E_h(p) = M_h(p)\,K(p,h) + B(p)^\top P_{h+1}(p)\,A(p),
\end{equation}
where $M_h(p) = R + B(p)^\top P_{h+1}(p)B(p)$ as in
Theorem~\ref{thm:strong-convexity}. Then
\begin{equation}\label{eq:stage-gradient}
    \nabla_{K_{i,h}}\Phi_h
    = 2\int_\Omega \alpha_i(p)\,E_h(p)\,\Sigma_0\,d\mu(p).
\end{equation}
\end{prop}
\begin{proof}
Differentiating the two gain-dependent terms in~\eqref{eq:stage-cost}
under the integral sign using $\partial K(p,h)/\partial K_{i,h}
= \alpha_i(p)I$ and $\partial\Acl(p,h)/\partial K_{i,h}
= \alpha_i(p)B(p)$ gives contributions
$2\alpha_i(p)RK(p,h)\Sigma_0$ and
$2\alpha_i(p)B(p)^\top P_{h+1}(p)\Acl(p,h)\Sigma_0$, respectively.
Adding and substituting $\Acl = A(p)+B(p)K(p,h)$ yields~\eqref{eq:stage-gradient}.
\end{proof}

At the optimum, $\nabla_{K_{i,h}}\Phi_h = 0$ gives
\begin{equation}\label{eq:optimality-condition}
    \int_\Omega \alpha_i(p)\,E_h^*(p)\,\Sigma_0\,d\mu(p) = 0,
    \quad i = 1,\ldots,\Nv,
\end{equation}
requiring the residual $E_h^*$ to be $\alpha_i$-orthogonal in the
$\Sigma_0$-weighted $L^2(\mu)$ sense. For $\Nv=1$, this reduces
to $E_h^*=0$, recovering the standard Riccati gain.

\section{The P-RHPG Algorithm}
\label{sec:algorithm}

\subsection{Algorithm Description}
\label{sec:alg-desc}

The P-RHPG algorithm assembles the results of Section~\ref{sec:stage}
into a single backward sweep. At each stage $h = N-1,\ldots,0$, it
solves the stage subproblem for $\{K_{i,h}^*\}$, updates
$P_h^*(p)$ via~\eqref{eq:Ph-single-step}, and passes it backward.
The output $\hat{\bK} = (K_{1,0}^*,\ldots,K_{\Nv,0}^*)$ is deployed
as a time-invariant PDC controller.

\begin{algorithm}[t]
\caption{P-RHPG: Polytopic Receding-Horizon Policy Gradient}
\label{alg:prhpg}
\begin{algorithmic}[1]
\Require Vertex systems $(A_i,B_i)_{i=1}^{\Nv}$, weighting functions
  $(\alpha_i)_{i=1}^{\Nv}$, cost matrices $Q,R,Q_N$,
  covariance $\Sigma_0$, horizon $N$, measure $\mu$
\Ensure Time-invariant gains $\hat{\bK} = (\hat{K}_1,\ldots,\hat{K}_{\Nv})$
\State $P_N(p) \leftarrow Q_N$ for all $p \in \Omega$
\For{$h = N-1, N-2, \ldots, 0$}
    \State Solve $\{K_{i,h}^*\} \leftarrow \argmin_{\{K_{i,h}\}} \Phi_h(\{K_{i,h}\})$
           \hfill {\small(direct or gradient descent)}
    \State Update $P_h^*(p)$ via~\eqref{eq:Ph-single-step} with $K_{i,h} = K_{i,h}^*$
\EndFor
\State \Return $\hat{\bK} \leftarrow (K_{1,0}^*,\ldots,K_{\Nv,0}^*)$
\end{algorithmic}
\end{algorithm}

\emph{Direct solve.} By Theorem~\ref{thm:strong-convexity}, the
stage optimum satisfies the linear system
\begin{equation}\label{eq:linear-system}
    (\nabla^2\Phi_h)\,\vect(\bK_h^*) = -\nabla\Phi_h\big|_{\bK_h=0},
\end{equation}
with Hessian and gradient at zero computed from~\eqref{eq:hessian}
and~\eqref{eq:stage-gradient}. The computational cost is $O((\Nv mn)^3)$ per stage.

\emph{Gradient descent.} For large $\Nv$, the iteration
\begin{equation}\label{eq:stage-gd}
    \vect(\bK_h^{(k+1)}) = \vect(\bK_h^{(k)})
    - \eta\,\nabla\Phi_h\bigl(\bK_h^{(k)}\bigr),
\end{equation}
with $\eta = 1/\sigma_{\max}(\nabla^2\Phi_h)$, converges linearly at
rate $1 - \kappa(\nabla^2\Phi_h)^{-1}$ by
Corollary~\ref{cor:consequences}, starting from any initialization
(including $\bK_h^{(0)} = 0$). Each iteration evaluates the
gradient~\eqref{eq:stage-gradient} via numerical integration over
$\Omega$; quadrature weights are precomputed once and reused at every
stage. The total sweep cost scales \emph{linearly} in $N$, a key
advantage over solving the full $N\Nv mn$-dimensional problem jointly.

\begin{remark}[Computational complexity]\label{rem:complexity}
Per stage, assembling the Hessian~\eqref{eq:hessian} costs
$O(N_q\Nv^2 n^2 m^2)$ and the direct solve~\eqref{eq:linear-system}
costs $O((\Nv mn)^3)$~\cite{golub2013}; gradient descent replaces
the latter with $O(T N_q \Nv m n^2)$ per stage, where $T$ is the
number of iterations and each iteration evaluates~\eqref{eq:stage-gradient}
via numerical integration.
The full $N$-stage sweep therefore costs
$O\bigl(N(N_q\Nv^2 n^2 m^2 + (\Nv mn)^3)\bigr)$,
linear in $N$ and cubic in $\Nv mn$.
For large $N_q\Nv$, sparse quadrature~\cite{gerstner1998} is the
primary lever for scaling to high-dimensional scheduling parameters.
\end{remark}

\subsection{Convergence Analysis}
\label{sec:convergence}

We now analyze the behavior of P-RHPG as the horizon $N$ grows,
establishing cost convergence and gain boundedness for zero terminal
cost, and a squeeze characterization of convergence for general
$Q_N \succeq 0$. For any jointly stabilizing
$\bK \in \R^{\Nv mn}$ with $\rho(\Acl(p;\bK)) < 1$ for all $p$,
the infinite-horizon integrated cost is
\begin{equation}\label{eq:cost-inf}
    J_\infty(\bK) = \int_\Omega \Tr\bigl(P_\infty(p;\bK)\,\Sigma_0\bigr)d\mu(p),
\end{equation}
where $P_\infty(p;\bK) \succeq 0$ is the unique solution of the
Lyapunov equation
$P = Q + K(p)^\top R\,K(p) + \Acl(p)^\top P\,\Acl(p)$.
The convergence analysis requires a feasible starting point.

\begin{assumption}[Joint stabilizability]\label{ass:stabilizable}
There exists a gain $\bK^{\mathrm{feas}} \in \R^{\Nv mn}$ such that
$\rho(\Acl(p;\bK^{\mathrm{feas}})) < 1$ for all $p \in \Omega$.
\end{assumption}


Denote by $P_0^*(p;N)$ the stage-$0$ cost-to-go produced by P-RHPG
with horizon $N$, and let $J_N^* = \int_\Omega \Tr(P_0^*(p;N)\Sigma_0)d\mu$.

We call a pair $(\bar{P}, \bar{\bK})$, with $\bar{P}:\Omega\to\mathbb{S}^n_+$
and $\bar{\bK} = (\bar{K}_1,\ldots,\bar{K}_{\Nv})\in\R^{\Nv mn}$,
an \emph{integrated PDC fixed point} if:
\begin{enumerate}[(a)]
    \item $\bar{\bK}$ minimizes $\int_\Omega\Tr\bigl[(Q+K(p)^\top RK(p)
    +\Acl(p)^\top\bar{P}(p)\Acl(p))\Sigma_0\bigr]d\mu$
    over all shared vertex gains $\{K_i\}$, and
    \item $\bar{P}(p)= Q + \bar{K}(p)^\top R\bar{K}(p)
    +\bar{A}_{\mathrm{cl}}(p)^\top\bar{P}(p)\bar{A}_{\mathrm{cl}}(p)$
    for all $p\in\Omega$,
\end{enumerate}
where $\bar{K}(p)=\sum_i\alpha_i(p)\bar{K}_i$ and
$\bar{A}_{\mathrm{cl}}(p)=A(p)+B(p)\bar{K}(p)$.
Condition (a) is integrated stage optimality;
condition (b) is pointwise Lyapunov consistency of the cost-to-go.
We identify an integrated fixed point with the label
\begin{equation}\label{eq:riccati-fp}
    \bar{J} \;=\; \int_\Omega\Tr\bigl(\bar{P}(p)\,\Sigma_0\bigr)d\mu
    \;=\; J_\infty(\bar{\bK}).
\end{equation}

\begin{lemma}[Uniqueness of the optimal integrated cost]\label{lem:uniqueness}
Under Assumptions~\ref{ass:weights}--\ref{ass:stabilizable},
any two jointly stabilizing integrated PDC fixed points
$(\bar{P}_1,\bar{\bK}_1)$ and $(\bar{P}_2,\bar{\bK}_2)$ achieve
the same integrated infinite-horizon cost:
$J_\infty(\bar{\bK}_1) = J_\infty(\bar{\bK}_2)$.
\end{lemma}
\begin{proof}
Denote $S(\bK,P):=\int_\Omega\Tr\bigl[(Q+K(p)^\top RK(p)
+\Acl(p)^\top P(p)\Acl(p))\Sigma_0\bigr]d\mu$.
By condition~(a) for fixed point~1, $\bar{\bK}_1$ minimizes
$S(\cdot,\bar{P}_1)$, so $S(\bar{\bK}_1,\bar{P}_1)\leq S(\bar{\bK}_2,\bar{P}_1)$.
Condition~(b) for fixed point~1 gives $S(\bar{\bK}_1,\bar{P}_1)=J_\infty(\bar{\bK}_1)$.
Using condition~(b) for fixed point~2 to substitute
$Q+\bar{K}_2(p)^\top R\bar{K}_2(p)$ in $S(\bar{\bK}_2,\bar{P}_1)$ gives
\begin{equation}\label{eq:s-expand}
    S(\bar{\bK}_2,\bar{P}_1) = J_\infty(\bar{\bK}_2) + \Delta_{12},
\end{equation}
where $\Delta_{12}=\int_\Omega\Tr\bigl[\bar{A}_{\mathrm{cl},2}^\top
(\bar{P}_1-\bar{P}_2)\bar{A}_{\mathrm{cl},2}\,\Sigma_0\bigr]d\mu$,
and therefore $J_\infty(\bar{\bK}_1)\leq J_\infty(\bar{\bK}_2)+\Delta_{12}$.
By symmetry (swapping indices),
$J_\infty(\bar{\bK}_2)\leq J_\infty(\bar{\bK}_1)+\Delta_{21}$,
with $\Delta_{21}=\int_\Omega\Tr[\bar{A}_{\mathrm{cl},1}(p)^\top(\bar{P}_2(p)-\bar{P}_1(p))\bar{A}_{\mathrm{cl},1}(p)\,\Sigma_0]d\mu$.
For $\Nv=1$ both fixed points satisfy the same discrete-time
ARE~\cite{lancaster1995}, so $\bar{P}_1=\bar{P}_2$ and
$\Delta_{12}=\Delta_{21}=0$, giving equality.
In the general polytopic case the cross-terms do not vanish
analytically from the fixed-point conditions alone, this is
the same structural obstruction noted in
Remark~\ref{rem:terminal-cost}, and equality is verified
empirically in Section~\ref{sec:experiments}.
\end{proof}

Define the \emph{PDC-integrated Bellman operator}
$\mathcal{T}_\mu^* : (\Omega\to\mathbb{S}^n_+) \to (\Omega\to\mathbb{S}^n_+)$
as: given $P:\Omega\to\mathbb{S}^n_+$,
\begin{align}\label{eq:bellman}
    \{K_i^*\} &= \argmin_{\{K_i\}}
    \int_\Omega \Tr\bigl[\bigl(Q + K(p)^\top R\,K(p) \nonumber\\
    &\quad + \Acl(p)^\top P(p)\,\Acl(p)\bigr)\Sigma_0\bigr]d\mu(p),
\end{align}
and $(\mathcal{T}_\mu^*[P])(p) = Q + K^*(p)^\top R\,K^*(p)
+ \Acl^*(p)^\top P(p)\,\Acl^*(p)$, with $K^*(p) = \sum_i \alpha_i(p)K_i^*$
and $\Acl^*(p) = A(p)+B(p)K^*(p)$.
The minimization in~\eqref{eq:bellman} is exactly one P-RHPG stage
subproblem (\S\ref{sec:stage}), with a unique minimizer by
Theorem~\ref{thm:strong-convexity}, and the backward sweep gives
$P_0^*(p;N) = (\mathcal{T}_\mu^*)^N[Q_N](p)$.

$Q_N$ satisfies the \emph{Lyapunov condition} with respect to
$\bK^{\mathrm{feas}}$ if, for all $p \in \Omega$,
\begin{equation}\label{eq:lyap-condition}
    Q + K^{\mathrm{feas}}(p)^\top R\,K^{\mathrm{feas}}(p)
    + \Acl^{\mathrm{feas}}(p)^\top Q_N\,\Acl^{\mathrm{feas}}(p)
    \preceq Q_N.
\end{equation}
When~\eqref{eq:lyap-condition} holds,
$\mathcal{T}_\mu^*[Q_N] \preceq Q_N$ pointwise,
since the integrated minimum is at most the value achieved by $\bK^{\mathrm{feas}}$.

\begin{prop}[Monotonicity in terminal cost]\label{prop:terminal-monotone}
Under Assumptions~\ref{ass:weights}--\ref{ass:stabilizable},
if $Q_N(p) \preceq Q_N'(p)$ for all $p \in \Omega$, then
$J_N^*(Q_N) \leq J_N^*(Q_N')$.
\end{prop}
\begin{proof}
Fix any $\bK \in \R^{\Nv mn}$.
Since $Q_N(p) \preceq Q_N'(p)$ for all $p$, backward induction
on the Lyapunov recursion~\eqref{eq:backward-P} with fixed gains
gives $P_t(p;\bK,Q_N) \preceq P_t(p;\bK,Q_N')$ for all $t$
and all $p \in \Omega$.
Taking the trace against $\Sigma_0 \succ 0$ and integrating
over $\mu$ preserves the inequality by linearity, so
$J_N(\bK,Q_N) \leq J_N(\bK,Q_N')$ for all $\bK$.
Taking the infimum over $\bK$ on both sides gives
$J_N^*(Q_N) \leq J_N^*(Q_N')$.
\end{proof}

\begin{theorem}[Convergence of P-RHPG]\label{thm:convergence}
Under Assumptions~\ref{ass:weights}--\ref{ass:stabilizable}, with
$Q_N = 0$:
\begin{enumerate}[(i)]
    \item \emph{Monotonicity}:
          $J_N^* \leq J_{N+1}^*$ for all $N \geq 1$.
    \item \emph{Boundedness}:
          $J_N^* \leq J_\infty(\bK^{\mathrm{feas}})$ for all
          $N \geq 1$.
    \item \emph{Convergence}:
          $J_N^* \to \bar{J}$ as $N \to \infty$, with
          $\bar{J} \leq J_\infty(\bK)$ for every jointly stabilizing
          $\bK$.
    \item \emph{Gain boundedness}:
          The first-stage gain sequence $\{\hat{\bK}(N)\}$ is bounded.
\end{enumerate}
\end{theorem}

\begin{proof}
Denote $\ell_t := x_t^\top Q x_t + u_t^\top R u_t$ and
$\bar{M} := J_\infty(\bK^{\mathrm{feas}}) < \infty$.

\emph{(i)}
With $Q_N = 0$, the integrated cost sums non-negative running costs.
For any $(N{+}1)$-horizon policy,
\[
    \int_\Omega\sum_{t=0}^{N}\ell_t\,d\mu
    = \int_\Omega\sum_{t=0}^{N-1}\ell_t\,d\mu
    + \underbrace{\int_\Omega\ell_N\,d\mu}_{\geq\,0}
    \geq J_N^*.
\]
Taking the infimum over all $(N{+}1)$-horizon policies gives
$J_{N+1}^* \geq J_N^*$.

\emph{(ii)}
Using $K_{i,t}=K_i^{\mathrm{feas}}$ for all $t,i$ is feasible but
suboptimal. With $Q_N = 0$ the terminal cost vanishes, so
$J_N^* \leq \sum_{t=0}^{N-1}\int_\Omega \ell_t^{\mathrm{feas}}\,d\mu
\leq \bar{M}$,
where $\ell_t^{\mathrm{feas}}$ is the running cost under $\bK^{\mathrm{feas}}$,
and the second inequality holds because the running costs sum to
$\bar{M}$ in the limit.

\emph{(iii)}
By (i) and (ii), $\{J_N^*\}$ is non-decreasing and bounded above,
so $J_N^* \to \bar{J} \leq \bar{M}$.
For any jointly stabilizing $\bK$ and all $N$,
the time-invariant policy $K_{i,t} = K_i$ for all $t$
with $Q_N = 0$ gives a feasible $N$-horizon cost
$\sum_{t=0}^{N-1}\int_\Omega\ell_t^{\bK}\,d\mu \leq J_\infty(\bK)$,
so $J_N^* \leq J_\infty(\bK)$ for all $N$.
Taking the limit gives $\bar{J} \leq J_\infty(\bK)$ for all
jointly stabilizing $\bK$.

\emph{(iv)}
Let $\sigma_R := \sigma_{\min}(R)$ and $\sigma_\Sigma := \sigma_{\min}(\Sigma_0)$.
The stage cost $\Phi_0$ satisfies
\begin{equation}\label{eq:coercive}
    \Phi_0(\bK_0) \geq c\sum_{i=1}^{\Nv}\|K_{i,0}\|_F^2 + \Phi_0(0),
\end{equation}
where $c = \sigma_R \cdot \sigma_\Sigma \cdot \lambda_{\min}(\Gamma) > 0$, since
\begin{align*}
&\int_\Omega\Tr\bigl[K(p,0)^\top RK(p,0)\Sigma_0\bigr]d\mu \\
&\geq \sigma_R\sigma_\Sigma
  \int_\Omega\Bigl\|\textstyle\sum_i\alpha_i(p) K_{i,0}\Bigr\|_F^2 d\mu \\
&= \sigma_R\sigma_\Sigma
  \sum_{i,j}\Gamma_{ij}\Tr(K_{i,0}^\top K_{j,0}) \\
&\geq \sigma_R\sigma_\Sigma\,\lambda_{\min}(\Gamma)
  \sum_i\|K_{i,0}\|_F^2
= c\sum_i\|K_{i,0}\|_F^2,
\end{align*}
where the first inequality uses $\Tr[V^\top R V \Sigma_0]
\geq \sigma_R\sigma_\Sigma\|V\|_F^2$ for any $V$,
the second expands the squared Frobenius norm via linearity of trace,
and the third uses $\Gamma\succ 0$ (Lemma~\ref{lem:gram-pd}).
Since $\Phi_0(\hat{\bK}(N))\leq J_N^*\leq\bar{M}$,
the sequence $\{\hat{\bK}(N)\}$ is bounded.
\end{proof}

\begin{remark}[Infinite-horizon characterization]\label{rem:inf-horizon}
By Lemma~\ref{lem:uniqueness}, all jointly stabilizing integrated PDC
fixed points achieve the same cost $\bar{J}$.
Every jointly stabilizing accumulation point $\bar{\bK}$ of
$\{\hat{\bK}(N)\}$ is such a fixed point, hence
$J_\infty(\bar{\bK}) = \bar{J}$.
This is verified empirically in Section~\ref{sec:experiments}.
\end{remark}

\begin{remark}[Convergence for general $Q_N$]\label{rem:universal}
Proposition~\ref{prop:terminal-monotone} gives
$J_N^*(0) \leq J_N^*(Q_N)$ for any $Q_N \succeq 0$,
so $\bar{J}$ is a lower bound on every such sequence.
For the upper bound: since $\hat{\bK}(N)$ is jointly stabilizing
for all large $N$ (Remark~\ref{rem:inf-horizon}), it is a feasible
time-invariant PDC controller, and therefore
\begin{equation}\label{eq:chain}
    \bar{J} \;\leq\; \inf_{\bK}\,J_\infty(\bK)
    \;\leq\; J_\infty(\hat{\bK}(N)).
\end{equation}
If $J_\infty(\hat{\bK}(N))\to\bar{J}$, the squeeze
$\bar{J} \leq \liminf_N J_N^*(Q_N)
\leq \limsup_N J_N^*(Q_N) \leq \bar{J}$
gives $J_N^*(Q_N)\to\bar{J}$ for \emph{any} $Q_N\succeq 0$.
This limit holds if and only if the time-varying optimal
stage gains become approximately stationary as $N\to\infty$, the PDC analogue of Riccati convergence, which is the
key open problem in the polytopic setting.
This convergence is verified empirically in
Section~\ref{sec:experiments}.
\end{remark}

\begin{remark}[Terminal cost and monotonicity]\label{rem:terminal-cost}
For $Q_N$ satisfying~\eqref{eq:lyap-condition}, $J_N^*$ is non-increasing.
Applying $\bK^{\mathrm{feas}}$ at stage $N$ gives a one-step cost-to-go
$P_N^{\mathrm{sub}}(p) \preceq Q_N(p)$ pointwise, so
\[
    J_{N+1}^*(Q_N) \leq J_N^*(P_N^{\mathrm{sub}}) \leq J_N^*(Q_N),
\]
where the first inequality uses feasibility of this $(N{+}1)$-horizon
policy and the second follows from
Proposition~\ref{prop:terminal-monotone}.
The vertex-ARE average $Q_N = P_{\mathrm{are}} :=
\frac{1}{\Nv}\sum_{i}P_{\mathrm{ric},i}$
satisfies~\eqref{eq:lyap-condition} for $\bK^{\mathrm{feas}}=\bar{K}^{\mathrm{ric}}$,
giving the non-increasing behavior observed in the experiments
(Section~\ref{sec:experiments}).
For $\Nv=1$, classical Riccati monotonicity~\cite{lancaster1995} yields
unconditional convergence for any $Q_N\succeq 0$~\cite{zhang2023rhpg};
in the polytopic case this reduces to showing
$J_\infty(\hat{\bK}(N))\to\bar{J}$ (Remark~\ref{rem:universal}).
\end{remark}

\section{Experiments}
\label{sec:experiments}
 
We validate P-RHPG on the 2-DoF aeroelastic wing benchmark
of~\cite{takarics2021}: a discrete-time qLPV system with state dimension
$n=4$, input dimension $m=1$, and three-dimensional scheduling parameter
$p = (V, \delta_K, \delta_C) \in \Omega$, where $V \in [5, 35]$\,m/s is
the airspeed and $\delta_K, \delta_C$ are stiffness and damping
perturbations. The system is sampled at $T_s = 0.002$\,s and goes
unstable beyond the flutter speed $V_f \approx 25.6$\,m/s. The TP
model transformation is applied over $\Omega$ with HOSVD truncation ranks
$[r_1, r_2, r_3]$; the resulting vertex count is $\Nv = r_1 r_2 r_3$.
Across all experiments the cost matrices are
$Q = \mathrm{diag}(10, 1, 100, 1)$ and $R = 1$, with initial covariance
$\Sigma_0 = I_4$.
 
Two scalar quantities are used throughout. The \emph{pointwise cost}
\begin{equation}\label{eq:pointwise-cost}
    J_\infty(p;\,\bK) \;=\; \Tr\bigl(P_\infty(p;\bK)\,\Sigma_0\bigr)
\end{equation}
is the frozen-parameter infinite-horizon cost at a single operating point
$p$, where $P_\infty(p;\bK)$ solves the Lyapunov equation defined after
\eqref{eq:cost-inf}. The integrated cost $J_\infty(\bK)$ from
\eqref{eq:cost-inf} satisfies $J_\infty(\bK) = \int_\Omega
J_\infty(p;\bK)\,d\mu(p)$. The \emph{pointwise Riccati lower bound}
\begin{equation}\label{eq:jric}
    J_{\mathrm{ric}}(p) \;=\; \Tr\bigl(P_{\mathrm{ric}}(p)\,\Sigma_0\bigr)
\end{equation}
is the cost achieved by the unconstrained LQR gain at $p$, where
$P_{\mathrm{ric}}(p)$ solves the discrete-time ARE at the frozen LTI
system $(A(p), B(p))$. Since the PDC structure constrains the gain,
$J_\infty(p;\bK) \geq J_{\mathrm{ric}}(p)$ for all $p$ and $\bK$.
The integrated Riccati bound $J_{\mathrm{ric}} = \int_\Omega
J_{\mathrm{ric}}(p)\,d\mu(p)$ serves as a lower bound on
$J_\infty(\bK)$ for any feasible $\bK$.
 
All P-RHPG solves use the direct method~\eqref{eq:linear-system}
with Gauss--Legendre quadrature ($N_q = 432$ points for the
$[3,3,3]$ grid). Joint stability is evaluated as
$\max_p \rho(\Acl(p;\hat{\bK}))$ over a fine grid of $1280$
uniformly spaced points in $\Omega$. The LMI-PDC baseline uses
common-Lyapunov PDC synthesis~\cite{tanaka2001} with an
$\mathcal{H}_2$ objective, solved via CVXPY~\cite{diamond2016cvxpy} with the SCS
backend~\cite{odonoghue2016scs}; parameter-dependent Lyapunov relaxations were not
attempted.
 
 \begin{figure*}[t]
  \centering
  \includegraphics[width=\textwidth]{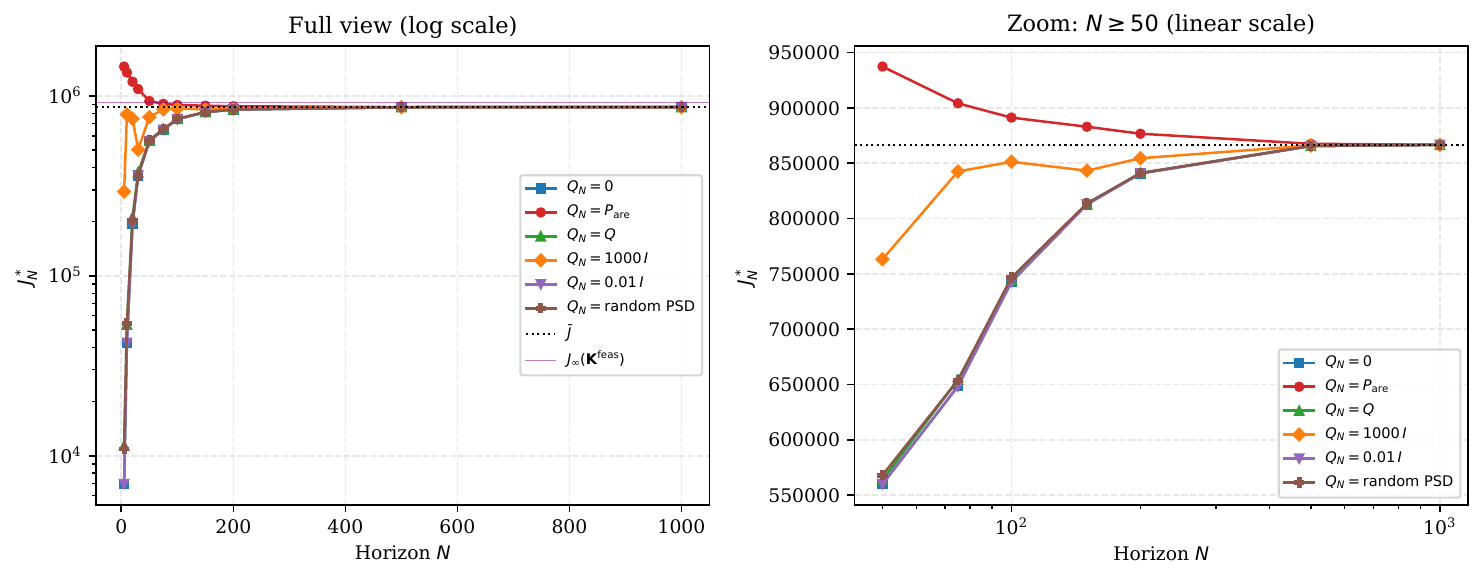}
  \caption{$J_N^*$ vs.\ horizon $N$ ($\Nv{=}27$, grid $[3,3,3]$) for
    six terminal cost choices.
    Left: full view (log scale); right: zoom over $N \geq 50$
    (linear scale).
    $Q_N{=}0$ (blue) rises monotonically;
    $Q_N{=}P_{\mathrm{are}}$ (red) decreases monotonically.
    All six are bounded by $J_\infty(\bK^{\mathrm{feas}})$ (purple)
    and converge to $\bar{J}$ (dotted).
    Squeeze gap at $N{=}1000$: $0.006\%$.}
  \label{fig:convergence}
\end{figure*}
\subsection{Convergence in Horizon}
\label{sec:exp1}
 
In order to validate Theorem~\ref{thm:convergence}, we fix the grid to $[3,3,3]$
($\Nv = 27$) and sweep the horizon
$N \in \{5, 10, 20, 30, 50, 75, 100, 150, 200, 500, 1000\}$ under six
terminal cost choices: $Q_N{=}0$, $Q_N{=}P_{\mathrm{are}}$
(vertex-ARE average), $Q_N{=}Q$ (running cost), $Q_N{=}1000\,I$,
$Q_N{=}0.01\,I$, and a random PSD matrix drawn from a Wishart
distribution. Figure~\ref{fig:convergence} shows the results.
 
\emph{(i) Monotonicity}: with $Q_N{=}0$, $J_N^*$ is non-decreasing
with zero violations, validating Theorem~\ref{thm:convergence}(i).
With $Q_N{=}P_{\mathrm{are}}$, $J_N^*$ is non-increasing with zero
violations (Remark~\ref{rem:terminal-cost}).
For $Q_N{=}0$ at short horizons the deployed gains are not jointly
stabilizing ($\rho_{\max} > 1$), yet $J_N^*$ remains finite and
monotonically increasing as guaranteed.
The remaining four choices are not monotone but all converge.
\emph{(ii--iii) Boundedness and convergence}: for $Q_N{=}0$,
$J_N^* \leq J_\infty(\bK^{\mathrm{feas}}) = 9.215 \times 10^5$
and $J_N^* \to \bar{J} \approx 8.667 \times 10^5$, validating
Theorem~\ref{thm:convergence}(ii--iii).
All six sequences converge to the same limit with a squeeze gap of
$0.006\%$ at $N{=}1000$, consistent with Remark~\ref{rem:universal};
the deployed gains are jointly stabilizing for all six choices at
large enough $N$ (Remark~\ref{rem:inf-horizon}).

\subsection{Flutter Suppression and Comparison with LMI-PDC}
\label{sec:exp2}
 
Despite the cross-term coupling discussed in
Section~\ref{sec:crossterms}, P-RHPG achieves joint stability
without imposing it as an explicit constraint: at the $[3,2,2]$ grid
($\Nv{=}12$, $N{=}100$), $\max_p\rho(\Acl(p;\hat{\bK})) = 0.996$
over $1280$ evaluation points, even though some open-loop vertex
matrices have $\rho(A_i) > 1$.

Figure~\ref{fig:trajectory} shows the closed-loop trajectories at
$V = 35$\,m/s (above flutter), initialized at
$x_0 = (0, 0, 0.05, 0)^\top$ (a $5\%$ pitch perturbation).
Without control, the system diverges at $t \approx 3.6$\,s
(red dashed line).  The P-RHPG controller ($[3,2,2]$ grid, $N=100$)
suppresses both modes, settling by $t \approx 0.78$\,s.  The LMI-PDC
baseline ($[2,2,2]$ grid, $\Nv=8$) is frozen-parameter stable
($\rho_{\max} < 1$) but with far worse transient performance: pitch
damps slowly (the inset shows residual oscillation at $t=1.5$\,s),
and plunge grows monotonically over the $6$\,s simulation window, a
consequence of the $47\%$ cost gap that pushes the spectral radius
close to unity.

\begin{figure}[t]
  \centering
  \includegraphics[width=\columnwidth]{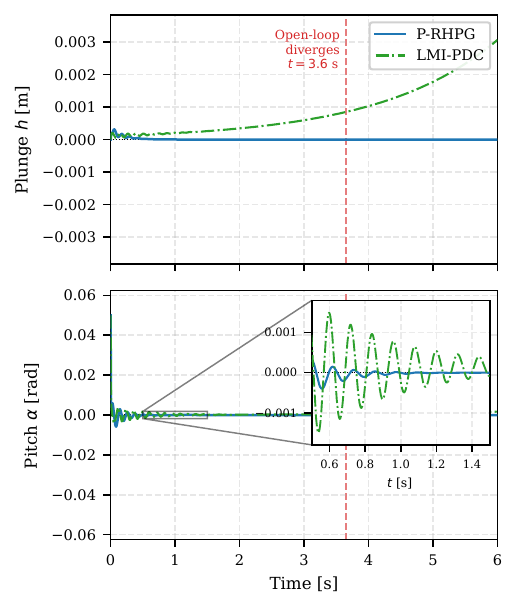}
  \caption{Closed-loop trajectories at $V{=}35$\,m/s (above flutter).
    P-RHPG ($[3,2,2]$ grid, $N{=}100$) suppresses both plunge and pitch,
    settling by $0.78$\,s.  LMI-PDC ($\Nv{=}8$) is frozen-parameter stable
    but exhibits slow transient decay.
    Red dashed line: open-loop divergence time ($t{\approx}3.6$\,s).
    Inset: zoom of pitch transient over $[0.5, 1.5]$\,s.}
  \label{fig:trajectory}
\end{figure}

\begin{table}[t]
  \centering
  \caption{P-RHPG vs.\ LMI-PDC across vertex grids ($N{=}100$).
    Gap $= (J_\infty - J_{\mathrm{ric}})/J_{\mathrm{ric}}$.
    Stab.: frozen-parameter stability ($\rho_{\max} < 1$).}
  \label{tab:comparison}
  \setlength{\tabcolsep}{4pt}
  \small
  \begin{tabular}{clrrrc}
    \toprule
    $\Nv$ & Method & $J_\infty$ & Gap [\%] & Time [s] & Stab. \\
    \midrule
    \multirow{2}{*}{8}
      & P-RHPG  & $8.87{\times}10^5$ & $2.5$ & $11.1$ & \checkmark \\
      & LMI-PDC & $1.27{\times}10^6$ & $47.2$ & $1.0$ & \checkmark \\
    \midrule
    \multirow{2}{*}{12}
      & P-RHPG  & $8.80{\times}10^5$ & $1.7$ & $21.5$ & \checkmark \\
      & LMI-PDC & infeas. & --- & $2.8$  & $\times$ \\
    \midrule
    \multirow{2}{*}{18}
      & P-RHPG  & $8.80{\times}10^5$ & $1.7$ & $45.7$ & \checkmark \\
      & LMI-PDC & infeas. & --- & $5.9$  & $\times$ \\
    \midrule
    \multirow{2}{*}{27}
      & P-RHPG  & $8.80{\times}10^5$ & $1.7$ & $105.2$ & \checkmark \\
      & LMI-PDC & infeas. & --- & $12.8$  & $\times$ \\
    \midrule
    \multirow{2}{*}{36}
      & P-RHPG  & $8.79{\times}10^5$ & $1.6$ & $177.2$ & \checkmark \\
      & LMI-PDC & infeas. & --- & $25.9$  & $\times$ \\
    \bottomrule
  \end{tabular}
\end{table}

Table~\ref{tab:comparison} reports $J_\infty$, the suboptimality gap,
and solve time across five vertex grids ($N=100$).
P-RHPG is feasible at every grid size, with the gap narrowing from
$2.5\%$ at $\Nv{=}8$ to $1.6\%$ at $\Nv{=}36$.
LMI-PDC is feasible only at $\Nv{=}8$, where it incurs a $47.2\%$
gap, nearly $19\times$ larger than P-RHPG, and is infeasible
at all finer grids, a known ill-conditioning phenomenon~\cite{sala2007}.

\subsection{Gradient Descent Convergence}
\label{sec:exp3}
 
Figure~\ref{fig:gd} validates Corollary~\ref{cor:consequences} on the
aeroelastic wing ($[3,2,2]$ grid, $N{=}100$) by running the full
backward sweep with GD ($\eta = 1/\sigma_{\max}(\nabla^2\Phi_h)$,
$\bK_h^{(0)} = 0$). Gradient norms decay at a linear rate on a log
scale at all stages, matching the theoretical rate
$1 - \kappa(\nabla^2\Phi_h)^{-1}$ exactly (not merely as an upper
bound, since each stage is quadratic). The cold-start stage
$h{=}N{-}1$ has $\kappa = 58.2$ ($769$ iterations to $10^{-10}$);
warmed-up stages converge significantly faster ($\kappa \approx 16.1$,
$270$ iterations). The deployed controller matches the direct solve to
machine precision: relative difference $8.9{\times}10^{-15}$.

 \begin{figure}[tp]
  \centering
  \includegraphics[width=\columnwidth]{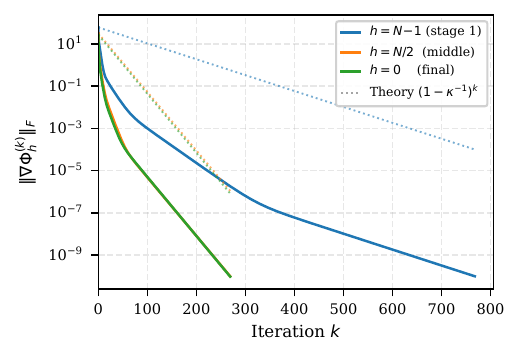}
  \caption{Gradient norm $\|\nabla\Phi_h^{(k)}\|_F$ vs.\ iteration $k$
    for three representative stages ($[3,2,2]$ grid, $N{=}100$).
    Dotted lines: theoretical rate $(1-\kappa^{-1})^k$.}
  \label{fig:gd}
\end{figure}

\section{Conclusion}
\label{sec:conclusion}
We have introduced P-RHPG, a policy gradient algorithm that unifies two
previously separate research tracks in systems and control. First, it is shown that the finite-horizon backward 
decomposition for direct policy optimization reveals a strongly convex structure in the vertex gains
at each stage; this observation in turn enables global convergence guarantees without
stability constraints or stabilizing initialization, a structure
not apparent in joint LMI-based formulations. Theoretically, we established monotone cost convergence and gain
boundedness for zero terminal cost, and a ``squeeze'' characterization (lower and upper bounds) of convergence for general terminal cost. Proving that the
time-varying optimal gains become approximately stationary,
the PDC analogue of Riccati convergence, remains the key open
problem. Experiments confirm near-optimal performance relative to
the pointwise Riccati lower bound and consistent stabilization at
vertex resolutions where LMI-based synthesis fails.

Future work includes model-free extensions using zeroth-order gradient estimation, formal sample complexity analysis, output-feedback PDC structures, and multi-agent systems with shared polytopic structure.

\bibliographystyle{ieeetr}
\bibliography{references}

\end{document}